\documentclass[11pt]{article}

\usepackage[margin=1in]{geometry}
\usepackage{amsmath,amssymb,amsthm,mathtools}
\usepackage{enumitem}
\usepackage{hyperref}
\usepackage{xcolor}                 % [added] for the corrections below
\hypersetup{
    colorlinks=true,
    linkcolor=blue,
    citecolor=blue,
    urlcolor=blue
}

\newtheorem{theorem}{Theorem}[section]
\newtheorem{proposition}[theorem]{Proposition}
\newtheorem{lemma}[theorem]{Lemma}
\newtheorem{corollary}[theorem]{Corollary}
\newtheorem{definition}[theorem]{Definition}
\newtheorem{remark}[theorem]{Remark}

\newcommand{\F}{\mathbb F}
\newcommand{\K}{\mathbb K}
\newcommand{\supp}{\operatorname{supp}}
\newcommand{\wt}{\operatorname{wt}}
\newcommand{\Span}{\operatorname{span}}
\newcommand{\RS}{\operatorname{RS}}

\newcommand{\Enc}{\operatorname{Enc}}
\newcommand{\Dec}{\operatorname{Dec}}
\newcommand{\enc}{\operatorname{enc}}
\newcommand{\dec}{\operatorname{dec}}

\title{Quantum Codes with Transversal \texorpdfstring{$CCZ$}{CCZ} Gates and Sublinear \texorpdfstring{$Z$}{Z}-Stabilizers}\author{
Ohad Elishco\thanks{
Ohad Elishco is with the School of Electrical and Computer Engineering,
Ben-Gurion University of the Negev, Israel
(e-mail: elishco@gmail.com).
} \ and Itzhak Tamo\thanks{
Itzhak Tamo is with the Department of Electrical Engineering--Systems,
Tel Aviv University, Tel Aviv 6997801, Israel (e-mail: tamo@tauex.tau.ac.il).
}
}
\date{}

\begin{document}
\maketitle

\begin{abstract}
We construct asymmetric quantum CSS codes with transversal \(CCZ\) gates from
algebraic expander codes \cite{KT26}.  For every fixed \(m\ge 3\), our growing-alphabet
codes have length \(N\), dimension \(\Theta(N)\), and  distances
\[
        d_X=\Theta(N),
        \qquad
        d_Z=\Theta(N^{1/m}).
\]
Moreover, the \(Z\)-stabilizer space has an explicit generating set of weight
\(O(N^{1/m})\).  

We build on the algebraic puncturing framework of Golowich and Guruswami \cite{GG24}, which turns classical codes with the required Schur-product and distance conditions into CSS codes with transversal \(CCZ\).
However, applying the framework directly to the algebraic expander codes  runs into their small dual distance, and therefore produces only sublinear dimension.
Our main technical step is a refined puncturing theorem in which the global dual-distance assumption is replaced by a condition only on the selected puncturing set.
 We also reduce the alphabet to a fixed prime field using a projective-multiplicity version of multiplication-friendly codes. 
The resulting fixed-prime-field CSS code triples, of length \(n\), still have transversal \(CCZ\) gates. Their
dimension is \(\Theta(n/(\log n)^4)\),  with distances 
\[
        d_X=\Omega\!\left(\frac{n}{(\log n)^4}\right),
        \qquad
        d_Z=\Omega\!\left(\frac{n^{1/m}}{(\log n)^{4/m}}\right),
\]
and the \(Z\)-stabilizer generating set remains sublinear.
\end{abstract}

\section{Introduction}
Transversal gates are one of the simplest mechanisms for implementing logical gates fault-tolerantly. Because they act independently on corresponding physical qudits, a single fault during a transversal gate cannot spread to many qudits within the same code block. At the same time, the Eastin--Knill theorem shows that, for any nontrivial quantum code detecting arbitrary single-subsystem errors, transversal gates cannot by themselves form a universal encoded gate set \cite{EastinKnill2009}. This makes it important to identify code families that support useful transversal non-Clifford gates while retaining strong coding parameters. One central example is the \(CCZ\) gate, a diagonal gate in the third level of the Clifford hierarchy \cite{CGK17}. Together with Clifford gates, \(CCZ\) gives a universal gate set.
 Constructions based on triorthogonal codes gave early examples of transversal non-Clifford gates, but do not yield asymptotically good code families \cite{BravyiHaah2012}.

Krishna and Tillich showed that algebraic multiplicative structure in classical codes can be used to construct quantum CSS codes with transversal non-Clifford gates \cite{KrishnaTillich2019}. In particular, using punctured Reed--Solomon (RS) codes, they obtained families supporting transversal third-level Clifford hierarchy gates with dimension and distance linear in the block length, over alphabets whose size grows with the block length. Golowich and Guruswami later isolated the general mechanism behind this construction: if a classical code \(C\), its dual \(C^\perp\), and the triple Schur-product code \(C^{\ast 3}\) all satisfy appropriate distance conditions, then puncturing a set of logical coordinates produces a CSS code with transversal \(CCZ\) \cite{GG24}. They instantiated this framework with algebraic-geometric codes and combined it with an alphabet-reduction procedure based on quantum multiplication-friendly codes, obtaining asymptotically good quantum codes over arbitrary fixed prime-power alphabets. 
A concurrent independent construction with similar asymptotically good parameters and applications to constant-overhead magic-state distillation was obtained by Wills, Hsieh, and Yamasaki \cite{WHY25}.
Related fixed-alphabet and binary constructions were subsequently given by Nguyen \cite{Nguyen24}. These results show that transversal \(CCZ\) gates can be realized by code families whose dimension and distance are both linear in the block length. What these algebraic constructions above do not provide, however, is low-weight stabilizer checks. This matters for syndrome extraction, since high-weight checks are harder to measure fault-tolerantly and usually require more complicated measurement circuits.  In this work we also keep track of the two CSS distances separately, since the construction parameters we obtain are highly asymmetric.

    Several recent works have focused directly on the weight of the stabilizer checks. Golowich and Lin constructed qLDPC codes with transversal \(C^{r-1}Z\) gates, almost-linear dimension, polynomial distance, and polylogarithmic stabilizer weight \cite{GolowichLin25}. Golowich and Guruswami subsequently obtained subsystem-code constructions with transversal \(CCZ\), linear dimension and distance, and sublinear parity-check locality \cite{GG25}. In their main construction the locality is \(O(\sqrt N)\) over an alphabet of size \(\Theta(\sqrt N)\), and the alphabet can be made constant at the cost of polylogarithmic losses. Another recent direction uses topological and homological constructions, based on cup and cap products, to obtain transversal \(CCZ\) and more general multi-controlled-\(Z\) gates for qLDPC and qLTC families \cite{Zhu25,LiLiLiu2026}.
 Taken together, these works show that low-weight checks are compatible with transversal non-Clifford gates. Their techniques, however, are quite different from the Schur-product-based algebraic puncturing framework that we use here.
    
    The point of comparison most relevant to the present paper is not that our
minimum-distance parameters dominate these constructions. Rather, we return to
the algebraic puncturing method underlying
\cite{KrishnaTillich2019,GG24} and ask whether it can produce stabilizer
locality.  We show that
it can, at least on one stabilizer side.

More precisely, we refine the Golowich--Guruswami puncturing argument and apply
the refined version to the algebraic expander codes of Kopparty and the second
author \cite{KT26}. For every fixed integer \(m\ge 3\), this gives
growing-alphabet CSS codes of block length \(N\), linear dimension, transversal
\(CCZ\), and asymmetric distances
\[
        d_X=\Theta(N),
        \qquad
        d_Z=\Theta(N^{1/m}).
\]
In addition, the \(Z\)-stabilizer space has an explicit generating set of weight
\(O(N^{1/m})\). These \(Z\)-stabilizers are the checks used to detect \(X\)-type
errors, so the low-weight checks are aligned with the linear \(X\)-distance.
Thus the construction is not qLDPC, and the low-weight generation is only
one-sided, but it shows that sublinear stabilizer locality is  possible within the algebraic Schur-product approach.

\paragraph{Beyond the Black-Box Puncturing Bound.}
Our construction uses the algebraic expander codes of \cite{KT26}. These are Tanner-type codes \cite{Tanner81} whose local constraints are RS codes supported on the orbits of two noncommuting subgroups of \(\operatorname{AGL}(1,\mathbb F)\). For every fixed local rate parameter \(r\in(0,1)\), these codes have positive rate and linear distance. The point that matters for us is the low-rate regime: the standard constraint-counting argument for Tanner codes gives only the lower bound \(2r-1\) on the global rate, and hence gives no nontrivial rate bound when \(r\le 1/2\). Being able to work below this threshold is crucial for Schur-product constructions. Indeed, if \(C_r\) denotes the algebraic expander code with local rate \(r<1/3\), then the product of three codewords of \(C_r\) lies in \(C_{3r}\), which has linear distance.

However, a direct application of the Golowich--Guruswami theorem to these codes does not give linear dimension. The obstruction is the small dual distance of the algebraic expander code \(C_r\): in the relevant parameter regime, one has \(d(C_r^\perp)=\Theta(N^{1/m})\). In the black-box theorem, the number of logical coordinates that can be punctured is bounded by this dual distance, so one obtains only \(O(N^{1/m})\) logical qudits. Thus, although the underlying classical code \(C_r\) has linear dimension, the resulting quantum code has vanishing rate.

Our main technical contribution is to show that this dimension loss is an artifact of the black-box dual-distance condition. We prove a refined puncturing theorem that replaces the global dual-distance condition by an \(A\)-dependent condition: what matters is how much support a dual codeword must have outside the chosen logical set \(A\). We then choose \(A\) to be an interpolation set for a lower-rate algebraic expander code \(C_{r_0}\), with \(r_0<r\). This choice gives \(|A|=\Theta(N)\), while still ensuring that every nonzero word in \(C_r^\perp\) has \(\Omega(N^{1/m})\) nonzero coordinates outside \(A\). Consequently, the punctured CSS code has linear dimension and asymmetric distances \(d_X=\Theta(N)\) and \(d_Z=\Theta(N^{1/m})\), rather than the sublinear dimension obtained from the black-box application.

\paragraph{Main Contributions.}
We summarize the main results of the paper.

First, we obtain an explicit asymmetric construction with one-sided locality
over growing alphabets.  For every fixed integer \(m\ge 3\) and every prime
characteristic \(p\), we construct CSS codes over fields of characteristic \(p\)
with length \(N\), linear dimension,  distances
\[
        d_X=\Theta(N),
        \qquad
        d_Z=\Theta(N^{1/m}),
\] and polynomial field size.  
The codes support transversal \(CCZ\) gates, and their \(Z\)-stabilizer space has
an explicit generating set of weight \(O(N^{1/m})\). The generators have a simple
description in terms of the Tanner structure of the underlying algebraic
expander code: they are supported either on a single orbit or on the union of
two orbits. Since these \(Z\)-stabilizers detect \(X\)-type errors, the
sublinear-check side is aligned with the linear-distance side.

Second, we reduce the alphabet to a fixed prime field. To do so, we use the alphabet-reduction framework of Golowich and Guruswami and track the effect of concatenation and restriction on the \(Z\)-stabilizer generators. For the inner codes, building on the multiplication-friendly-code paradigm \cite{GJX17}, we introduce a projective-multiplicity version over \(\mathbb P^1(\mathbb F_q)\), which works over arbitrary extensions and in arbitrary characteristic. This gives, for every fixed prime \(p\), CSS code triples over \(\mathbb F_p\) supporting transversal \(CCZ\), with parameters
\[
\left[\left[
n,\
\Omega\left(\frac{n}{(\log n)^4}\right),\
        d_X=\Omega\!\left(\frac{n}{(\log n)^4}\right), d_Z=\Omega\!\left(\frac{n^{1/m}}{(\log n)^{4/m}}\right)  
\right]\right]_p,
\]
and with \(Z\)-stabilizer generators of weight
\(
O\left(n^{1/m}(\log n)^{1-4/m}\right).
\)

\paragraph{Organization.}
Section~\ref{sec:preliminaries} recalls the necessary background on CSS codes, Schur products, transversal \(CCZ\) gates, and algebraic expander codes. Section~\ref{sec:puncturing-theorem} proves the refined puncturing theorem. Section~\ref{sec:aec-input} applies this theorem to algebraic expander codes and proves the growing-alphabet construction, including the explicit low-weight generating set for the \(Z\)-stabilizers. Section~\ref{sec:one-shot-alphabet-reduction} gives the fixed-prime-field alphabet reduction using projective-multiplicity multiplication-friendly codes. Section~\ref{sec:discussion} concludes with open problems.

\section{Preliminaries}
\label{sec:preliminaries}

Throughout, all vector spaces are over a finite field unless explicitly stated otherwise.
For a finite set \(\Omega\), we write \(\F^\Omega\) for the space of functions \(\Omega\to\F\). For \(x\in\F^\Omega\), let \(\supp(x)=\{\omega\in\Omega:x_\omega\ne 0\}\) and \(\wt(x)=|\supp(x)|\). The standard bilinear form on \(\F^\Omega\) is \(\langle x,y\rangle=\sum_{\omega\in\Omega}x_\omega y_\omega\). For a linear code \(C\subseteq\F^\Omega\), its dual is \(C^\perp=\{x\in\F^\Omega:\langle x,c\rangle=0\text{ for all }c\in C\}\).
If \(S\subseteq\Omega\), we write \(C|_S\) for the puncturing of \(C\) to \(S\). We also identify \(\F^S\) with the subspace of \(\F^\Omega\) consisting of vectors supported on \(S\). Thus \(C\cap \F^S\) is the shortening of \(C\) to \(S\), viewed inside \(\F^\Omega\). We will repeatedly use the standard
puncturing-shortening dualities
\begin{equation}
        (C|_S)^\perp=(C^\perp\cap \F^S)|_S,
        \qquad
        (C^\perp|_S)^\perp=(C\cap \F^S)|_S.
        \label{eq:puncture-shorten-duality}
\end{equation}

For a nonzero linear code \(C\subseteq \F^\Omega\), its minimum distance is
\[
        d(C)=\min_{0\ne c\in C}\wt(c).
\]
 If \(|\Omega|=n\), \(\dim_\F C=k\), and \(d(C)=d\),
we say that \(C\) is an \([n,k,d]_{\F}\) code, or an \([n,k,d]_q\) code when
\(|\F|=q\).

Let \(\F[x]^{<k}\) be the space of polynomials of degree less than \(k\) over \(\F\).
\paragraph{Interpolation sets, Schur powers, and RS codes}

For \(A\subseteq\Omega\), we say that \(A\) is an interpolation set for \(C\) if \(C|_A=\F^A\). Equivalently, every assignment of values on \(A\) can be interpolated by a codeword of \(C\).

For \(x,y\in\F^\Omega\), their Schur product is \(x*y=(x_\omega y_\omega)_{\omega\in\Omega}\). For linear codes \(C_1,\ldots,C_t\subseteq\F^\Omega\), define \(C_1*\cdots*C_t=\Span_\F\{c_1*\cdots*c_t:c_i\in C_i\}\). When all \(C_i=C\), we write \(C^{*t}=C*\cdots*C\). The case \(t=3\) is the one relevant for transversal \(CCZ\) gates.

For a finite set \(B\subseteq\F\) and an integer \(k\le |B|\), define the Reed-Solomon code
\[
        \RS_B(k)=\{(f(\beta))_{\beta\in B}:\deg f<k\}\subseteq\F^B.
\]
This is a \([|B|,k,|B|-k+1]_\F\) MDS code. In particular, \(d(\RS_B(k)^\perp)=k+1\).

\subsection{CSS codes and restriction}

We use the standard CSS formalism, with notation for specified logical
encoding maps following \cite{GG24}. 

A CSS code over \(\F\) is a pair
\(Q=\operatorname{CSS}(Q_X,Q_Z)\) with \(Q_X,Q_Z\subseteq\F^n\) such that
\(Q_X^\perp\subseteq Q_Z\).
Its dimension is
\(
        \dim Q=\dim Q_Z-\dim Q_X^\perp .
\)
Elements of \(Q_X^\perp\) are called \(X\)-stabilizers, and elements of
\(Q_Z^\perp\) are called \(Z\)-stabilizers.

We will keep track of the two CSS distances separately. With the above convention,
the minimum weight of a nontrivial \(X\)-type logical operator is
\[
        d_X(Q)
        =
        \min\{\wt(x):x\in Q_Z\setminus Q_X^\perp\},
\]
and the minimum weight of a nontrivial \(Z\)-type logical operator is
\[
        d_Z(Q)
        =
        \min\{\wt(z):z\in Q_X\setminus Q_Z^\perp\}.
\]
The usual quantum distance is
\(
        d(Q)=\min\{d_X(Q),d_Z(Q)\}.
\)
When we want to keep track of the two distances, we will write the parameters as
\([[n,k,d_X,d_Z]]_\F\).

A \(Z\)-encoding map for \(Q\) is an \(\F\)-linear isomorphism \(\Enc_Z:\mathbb F^{\dim Q}\xrightarrow{\sim} Q_Z/Q_X^\perp\), where \(\mathbb F^{\dim Q}\) is the logical space. We write \(Q=\operatorname{CSS}(Q_X,Q_Z;\Enc_Z)\) when the \(Z\)-encoding map is part of the data.

We use the same notation \(\Enc_Z(y)\) both for
the quotient class in \(Q_Z/Q_X^\perp\) and for the corresponding coset inside
\(Q_Z\). Thus \(z\in\Enc_Z(y)\) means that \(z\in Q_Z\) is a representative of
the logical \(Z\)-coset labeled by \(y\).

\paragraph{Compatible \(X/Z\)-encodings.}
 An \(X\)-encoding map for
\(Q=\operatorname{CSS}(Q_X,Q_Z)\), with logical space \(V\), is an \(\F\)-linear
isomorphism
\[
        \Enc_X:V\xrightarrow{\sim} Q_X/Q_Z^\perp .
\]
Given a nondegenerate bilinear form \(B:V\times V\to\F\), we say that
\((\Enc_X,\Enc_Z)\) are compatible with respect to \(B\) if, for every
\(x,z\in V\) and every choice of representatives
\(x'\in\Enc_X(x)\) and \(z'\in\Enc_Z(z)\), one has
\[
        B(x,z)=\langle x',z'\rangle .
\]
We write
\(
        Q=\operatorname{CSS}(Q_X,Q_Z;\Enc_X,\Enc_Z)
\)
when both encoding maps are given.

We will need the following lemma from \cite{GG24}.

\begin{lemma}[Compatible \(X\)-encodings {\cite[Lemma~2.5]{GG24}}]
\label{lem:compatible-X-encoding}
Let \(Q=\operatorname{CSS}(Q_X,Q_Z;\Enc_Z)\) be a CSS code with logical
space \(V\), and let \(B:V\times V\to\F\) be a nondegenerate bilinear form. Then there exists
an \(X\)-encoding map \(\Enc_X\) compatible with \(\Enc_Z\) with respect to \(B\).
\end{lemma}

\begin{remark}[Meaning of compatibility]
The compatibility condition says that the physical representatives chosen for
logical phase operators act on the encoded computational basis with the correct
logical phases. For example, in the binary case, the logical operator \(Z(y)\)
acts on the logical basis state \(|v\rangle\) by the phase
\((-1)^{y\cdot v}\). If \(x'\in\Enc_X(y)\) and \(z'\in\Enc_Z(v)\), then the
physical operator \(Z(x')\) acts on the basis vector \(|z'\rangle\) by the phase
\((-1)^{x'\cdot z'}\). Thus \(Z(x')\) implements logical \(Z(y)\) precisely when
\(x'\cdot z'=y\cdot v\). This is the compatibility condition with the standard
inner product as the logical bilinear form. Later, in the CSS concatenation
step, the inner codes will use the bilinear form
\(B(x,z)=\operatorname{Tr}_{\K/\F}(xz)\), which is why we allow an arbitrary
nondegenerate bilinear form \(B\).
\end{remark}

\paragraph{Code restriction.} 
We will also use the following restriction operation for CSS codes. Informally,
restriction adds \(Z\)-stabilizers, thereby reducing the logical dimension while
not decreasing the distance.

To define the operation we will need the following notation. For a subset \(S\subseteq V\), we write
\[
    \Enc_Z(S)=\{\Enc_Z(s):s\in S\}\subseteq Q_Z/Q_X^\perp .
\]

\begin{lemma}[Restriction of logical subspaces]
\label{lem:css-restriction}
Let \(Q=\operatorname{CSS}(Q_X,Q_Z;\Enc_Z)\) be a CSS code with logical space \(V\), and let \(S\subseteq V\) be a linear subspace. Define \(Q_Z(S)=\{z\in Q_Z:z+Q_X^\perp\in \Enc_Z(S)\}\). 
Then \(Q|_S=\operatorname{CSS}(Q_X,Q_Z(S);\Enc_Z|_S)\) is a CSS code of dimension \(\dim S\), and
\[
        d_X(Q|_S)\ge d_X(Q),\qquad
        d_Z(Q|_S)\ge d_Z(Q).
\]
In particular \(d(Q|_S)\ge d(Q)\). Thus restriction only removes logical degrees of freedom; equivalently, it adds \(Z\)-stabilizers.
\end{lemma}

\begin{proof}
Since \(Q_Z(S)\subseteq Q_Z\) and \(Q_X^\perp\subseteq Q_Z(S)\), the pair
\(\operatorname{CSS}(Q_X,Q_Z(S))\) is a CSS code. The quotient
\(Q_Z(S)/Q_X^\perp\) is identified with \(\Enc_Z(S)\), so the dimension is \(\dim S\).

For the distances, since \(Q_Z(S)\subseteq Q_Z\), we have
\(Q_Z(S)^\perp\supseteq Q_Z^\perp\). Hence
\[
        Q_Z(S)\setminus Q_X^\perp
        \subseteq
        Q_Z\setminus Q_X^\perp,
\]
so the \(X\)-distance cannot decrease. Similarly,
\[
        Q_X\setminus Q_Z(S)^\perp
        \subseteq
        Q_X\setminus Q_Z^\perp,
\]
so the \(Z\)-distance cannot decrease. Thus both CSS distances are nondecreasing under restriction.
\end{proof}

\subsection{Transversal \texorpdfstring{\(CCZ\)}{CCZ}}

Let \(\K\) be a finite field of characteristic \(p\). For \(a\in\K\), define the physical gate
\[
        CCZ_\K^a
        =
        \sum_{x,y,z\in\K}
        \exp\!\left(
            \frac{2\pi i}{p}
            \operatorname{Tr}_{\K/\F_p}(a xyz)
        \right)
        |x,y,z\rangle\langle x,y,z|.
\]
When \(a=1\), we write \(CCZ_\K=CCZ_\K^1\).

We will need the following definition. 
\begin{definition}
\label{def:transversal-ccz}
Let \(I\) be a finite index set and let the logical space be \(V=\K^I\). For \(h=1,2,3\), let \(Q^{(h)}=\operatorname{CSS}(Q_X^{(h)},Q_Z^{(h)};\Enc_Z^{(h)})\) be CSS codes over \(\K\), all of length  \(n\) and  logical space \(\K^I\). We say that the triple \((Q^{(1)},Q^{(2)},Q^{(3)})\) supports transversal \(CCZ_\K\)
if there exists a coefficient vector \(b\in\K^n\) such that for all
\(
        y^{(1)},y^{(2)},y^{(3)}\in\K^I
\)
and all representatives
\(
        z^{(h)}\in \Enc_Z^{(h)}(y^{(h)})\subseteq Q_Z^{(h)},
\)
one has
\begin{equation}
        \sum_{i\in I} y_i^{(1)}y_i^{(2)}y_i^{(3)}
        =
        \sum_{j=1}^n b_j z_j^{(1)}z_j^{(2)}z_j^{(3)}.
        \label{eq:transversal-ccz-identity}
\end{equation}
If \(Q^{(1)}=Q^{(2)}=Q^{(3)}\), we simply say that \(Q\) supports transversal \(CCZ_\K\).
\end{definition}
The next lemma shows that Definition~\ref{def:transversal-ccz} has the intended operational meaning. Namely, if
the identity in \eqref{eq:transversal-ccz-identity} holds for all logical labels and all
choices of physical representatives, then the physical transversal operator
\(
        \bigotimes_{j=1}^n CCZ_\K^{b_j}
\)
acts on the encoded computational basis exactly as the logical \(CCZ_\K\) gate on the logical
space \(\K^I\).
\begin{lemma}[{\cite[Lemma~2.11]{GG24}}]
\label{lem:ccz-identity-implements-gate}
For \(h=1,2,3\), let \(Q^{(h)}=\operatorname{CSS}(Q_X^{(h)},Q_Z^{(h)};\Enc_Z^{(h)})\) be CSS codes that support transversal \(CCZ_\K\) with coefficient vector \(b\in\K^n\). Then applying
\(
    \bigotimes_{j=1}^n CCZ_\K^{b_j}
\)
to the three physical code blocks implements the logical \(CCZ_\K\) gate on the 
logical basis indexed by \(\K^I\).
\end{lemma}

\begin{proof}
For logical basis states \(y^{(1)},y^{(2)},y^{(3)}\), the \(Z\)-encoded states are uniform
superpositions over the corresponding cosets
\(
        \Enc_Z^{(h)}(y^{(h)})\subseteq Q_Z^{(h)}.
\)
Consider \(3\) representatives  \(z^{(h)}\in \Enc_Z^{(h)}(y^{(h)})\).  Then the physical transversal gate multiplies the basis vector
\((z^{(1)},z^{(2)},z^{(3)})\) by the phase
\[
        \exp\!\left(
            \frac{2\pi i}{p}
            \operatorname{Tr}_{\K/\F_p}
            \left(\sum_{j=1}^n b_j z_j^{(1)}z_j^{(2)}z_j^{(3)}\right)
        \right)=  \exp\!\left(
            \frac{2\pi i}{p}
            \operatorname{Tr}_{\K/\F_p}
            \left(\sum_{i\in I} y_i^{(1)}y_i^{(2)}y_i^{(3)}\right)
        \right),
\]
where the equality follows by \eqref{eq:transversal-ccz-identity}. This is precisely the logical \(CCZ_\K\) phase.
\end{proof}

\subsection{Algebraic expander codes}
We now provide  the needed notation and results of the algebraic expander code construction given in \cite{KT26}.

Let \(u(X)\in\F[X]\) be nonconstant. Every polynomial \(f\in\F[X]\) has a unique expansion \(f(X)=\sum_i c_i(X)u(X)^i\) with \(\deg c_i<\deg u\). The \(u\)-base degree of \(f\) is \(\deg_u(f)=\max_i \deg c_i\), with the convention \(\deg_u(0)=-\infty\).

We will need the following result from \cite{KT26}.
\begin{lemma}[Subadditivity of base degree {\cite[Lemma~2.2]{KT26}}]
\label{lem:base-degree-subadditivity}
For all \(f_1,\ldots,f_t\in\F[X]\),
\[
        \deg_u(f_1\cdots f_t)
        \le
        \deg_u(f_1)+\cdots+\deg_u(f_t).
\]
\end{lemma}

\begin{proof}
If one of the \(f_i\)'s is zero, the claim is immediate. It is enough to prove the case \(t=2\).
Let \(d=\deg u\). If \(\deg_u(f)+\deg_u(g)\ge d\), then the claim  follows  because every
\(u\)-base degree is at most \(d-1\). Otherwise, write \(f=\sum_i c_i u^i\) and \(g=\sum_j b_j u^j\), where \(\deg c_i\le \deg_u(f)\) and \(\deg b_j\le \deg_u(g)\). Then \(fg=\sum_k e_k u^k\), where \(e_k=\sum_{i+j=k}c_i b_j\). Each \(e_k\) has degree at most \(\deg_u(f)+\deg_u(g)<\deg u\), so this is already the
\(u\)-base expansion of \(fg\). Hence
\[
        \deg_u(fg)\le \deg_u(f)+\deg_u(g).
\]
\end{proof}
We use the algebraic expander code construction of \cite{KT26} in the case where the
global degree parameter equals the block length. The construction is built from two
subgroups of \(\operatorname{AGL}(1,\F)\): a translation subgroup
\[
        G=\{x\mapsto x+a:a\in U\}
\]
for an additive subgroup \(U\le(\F,+)\), and a scaling subgroup
\[
        H=\{x\mapsto bx:b\in B\}
\]
for a multiplicative subgroup \(B\le\F^\times\). Let
\(\mathcal A=\langle G,H\rangle\) be the group generated by \(G,H\), which is a subgroup of the affine general linear group 
\(\operatorname{AGL}(1,\F)\). Let \(\alpha\in\F\) be an element with trivial stabilizer in \(\mathcal A\), meaning
that no non-identity element of \(\mathcal A\) fixes \(\alpha\), and set
\(
        \Omega=\{\phi(\alpha):\phi\in\mathcal A\}.
\)
Thus \(\Omega\) is a regular \(\mathcal A\)-orbit, i.e., \(\mathcal A\) acts
transitively on \(\Omega\) and every point stabilizer is trivial. In particular,
\(
        N:=|\Omega|=|\mathcal A|.
\)

Define the invariant polynomials associated with the two  group actions
\[
        g(X)=\prod_{a\in U}(X-a),
        \qquad
        h(X)=X^{|H|}.
\]
It is easy to verify that the polynomial \(g\) is constant on \(G\)-orbits, and \(h\) is constant on
\(H\)-orbits. For \(r\in(0,1)\), the algebraic expander code is
\begin{equation}
\label{eq:ex-const}
        C_r=C(G,H;r,N)
        =
        \{(f(\omega))_{\omega\in\Omega}:f\in W_{r,N}\},
\end{equation}
where
\[
        W_{r,N}
        =
        \{f\in\F[X]:
            \deg f<N,\ 
            \deg_g(f)<r|G|,\ 
            \deg_h(f)<r|H|
        \}.
\]
For an element \(\omega\in\Omega\), let
\(G\cdot\omega=\{g(\omega):g\in G\}\) be the orbit of \(\omega\) under \(G\),
and similarly define \(H\cdot\omega\). Let \(\mathcal O\) be the collection of all
\(G\)-orbits and \(H\)-orbits in \(\Omega\). 
For \(O\in\mathcal O\), let \(R_O(r)\subseteq\F^O\) be the RS code
\[
        R_O(r)=\{(P(\omega))_{\omega\in O}:P\in\F[X],\ \deg P<k_O(r)\},
        \qquad
        k_O(r)=\lceil r|O|\rceil .
\]
Then \(R_O(r)\) has dimension \(k_O(r)\), and
\(
        \{c|_O:c\in C_r\}\subseteq R_O(r)       
\) for every \(O\in\mathcal O\).
In fact, \(C_r\) is the Tanner code defined by these local codes \(R_O(r),O\in \mathcal{O}\).

The properties of \(C_r\) are summarized
in the next theorem.
\begin{theorem}[{\cite[Theorem~4.3]{KT26}}]
\label{prop:principal-aec-tanner}
For any fixed \(m\ge 3\), \(r\in(0,1)\), and characteristic \(p\), there exists an explicit
family of algebraic expander codes \(C_r\) of length \(N\), defined over finite fields
\(\F_q\) of characteristic \(p\), as in \eqref{eq:ex-const}, with alphabet size
\(
    q
\) polynomial in \(N\),
and with the following properties as \(N\to\infty\):
\begin{enumerate}
    \item \(C_r\) is a Tanner code. Specifically,
    \[
        C_r
        =
        \{c\in\F_q^\Omega:c|_O\in R_O(r)\text{ for every }O\in\mathcal O\}.
    \]
    Consequently,
    \begin{equation}
        C_r^\perp=\sum_{O\in\mathcal O}R_O(r)^\perp,
        \label{eq:aec-dual-local-generation}
    \end{equation}
    where each \(R_O(r)^\perp\) is extended by zero outside \(O\).

    \item The local codes \(R_O(r)\) have length \(\Theta(N^{1/m})\), local rate at most
    \(r+o(1)\), and local relative distance at least \(1-r-o(1)\).

    \item The global relative distance is at least \((1-r)^2\).

    \item The global rate is at least
    \(
        \frac{r^{2m+1}}{(2m+1)!}.
    \)
\end{enumerate}
\end{theorem}

\section{Puncturing theorem for transversal \texorpdfstring{$CCZ$}{CCZ}}
\label{sec:puncturing-theorem}
Our goal is to construct CSS codes with good parameters, supporting
transversal \(CCZ\) gates, and with low-weight \(Z\)-stabilizer generators.

Toward this goal, we give a mild refinement of the construction of
\cite[Theorem~3.1]{GG24}, which is itself an abstract version of the
construction in \cite{KrishnaTillich2019}. The proof is essentially the same as the proof of
\cite[Theorem~3.1]{GG24}, but we relax one of its assumptions. In particular,
the original theorem uses the condition \(|A|<d(C^\perp)\) to control the
distance after puncturing. We replace this global dual-distance condition by an
\(A\)-dependent distance condition. This is important for our application to
algebraic expander codes: although \(d(C^\perp)\) is only of order the local
orbit size, we can choose a much larger set \(A\) for which the relevant
punctured distance outside \(A\) remains large. This allows us to obtain a
linear-size logical set \(A\), and hence linear code dimension, whereas a
black-box application of \cite[Theorem~3.1]{GG24} would only give sublinear
dimension.

We will need the following notation.
Let $C\subseteq\F^\Omega$ be a linear code and let $A\subseteq\Omega$. Put $E=\Omega\setminus A$. 
Define the \(A\)-punctured distance
\[
    d_A(C)
    =
    \min\{|\supp(u)\setminus A|:u\in C,\ u|_A\ne 0\}.
\]
This quantity measures the part of a codeword that remains visible after the coordinates in \(A\) are removed, but only for codewords that actually interact with the logical coordinates in \(A\). This is the distance parameter that is relevant for the punctured CSS code, rather than the global minimum distance of the dual code.

\begin{theorem}[Refined puncturing theorem]\label{thm:refined-GG}
Let $C\subseteq\F^\Omega$ be a linear code and let $A\subseteq\Omega$. Suppose that
\begin{enumerate}[label=(\roman*)]
    \item $C|_A=\F^A$;
    \item $C\cap\F^A=\{0\}$;
    \item $C^{*3}\cap\F^A=\{0\}$.
\end{enumerate}
Let \(E=\Omega\setminus A\) and define \(Q_X=C^\perp|_E\) and \(Q_Z=C|_E\). Then \(Q=\operatorname{CSS}(Q_X,Q_Z)\) is a CSS code of length \(|E|\)
and dimension \(|A|\). Its two CSS distances are
\(
        d_X(Q)=d_A(C),       
        d_Z(Q)=d_A(C^\perp).
\)
In particular, its quantum distance is
\(
        d(Q)=\min\{d_A(C),d_A(C^\perp)\}.
\)
Moreover, \(Q\) supports a transversal \(CCZ\) gate over \(\F\).
\end{theorem}

\begin{proof}
By \eqref{eq:puncture-shorten-duality} we have 
\[
    Q_X^\perp=(C^\perp|_E)^\perp=(C\cap\F^E)|_E\subseteq C|_E=Q_Z,
\]
so the CSS condition holds.

Define the $Z$-encoding map for  \(y\in\F^A\) by
\[
    \operatorname{Enc}_Z(y)=\{c|_E:c\in C,\ c|_A=y\}, \qquad \operatorname{Enc}_Z:\F^A\rightarrow Q_Z/Q_X^\perp.
\]
The assumption $C|_A=\F^A$ ensures that \(\operatorname{Enc}_Z(y)\) is nonempty for every $y$. If two codewords $c,c'\in C$ have the same restriction to $A$, then $(c-c')|_A=0$, so $(c-c')|_E\in Q_X^\perp$. Hence $\operatorname{Enc}_Z$ is well defined as a map to $Q_Z/Q_X^\perp$. It is surjective by definition. If $\operatorname{Enc}_Z(y)$ is the zero coset, then there are $c\in C$ with $c|_A=y$ and $c'\in C\cap\F^E$ with $c'|_E=c|_E$. Thus $c-c'\in C\cap\F^A$ and $(c-c')|_A=y$. By assumption (ii), $y=0$. Hence the map is injective, and the CSS dimension is $|A|$.

We next compute the two CSS distances. A nontrivial \(X\)-type logical operator is
represented by an element of \(Q_Z\setminus Q_X^\perp\). Such an element has the
form \(c|_E\) for some \(c\in C\). It is trivial modulo \(Q_X^\perp\) exactly when
\(c|_A=0\). Indeed, if \(c|_A=0\), then \(c\in C\cap\F^E\), so
\(c|_E\in Q_X^\perp\). Conversely, if \(c|_E\in Q_X^\perp\), then there is
\(c'\in C\cap\F^E\) with \(c'|_E=c|_E\). Thus \(c-c'\in C\cap\F^A\), and
assumption (ii) forces \(c|_A=0\). Therefore
\[
        d_X(Q)=d_A(C).
\]

Similarly, assumption (i) implies \(C^\perp\cap\F^A=\{0\}\). Applying the same
argument to \(C^\perp\), a nontrivial \(Z\)-type logical operator is represented
by an element of \(Q_X\setminus Q_Z^\perp\), and its minimum possible weight is
\[
        d_Z(Q)=d_A(C^\perp).
\]
Thus the distance is
\(
        d(Q)=\min\{d_X(Q),d_Z(Q)\}
        =
        \min\{d_A(C),d_A(C^\perp)\}.
\)

It remains to prove the transversal $CCZ$ property. Since $C^{*3}\cap\F^A=\{0\}$, restriction to $E$ is injective on $C^{*3}$. Therefore the mapping
\[
    c|_E\longmapsto \sum_{a\in A} c_a,
    \qquad c\in C^{*3},
\]
defines a linear functional on $C^{*3}|_E$. Extend this functional arbitrarily to a linear functional on all of $\F^E$, and write it as $x\mapsto \sum_{e\in E} b_e x_e$ for some $b\in\F^E$.

Now take $y^{(1)},y^{(2)},y^{(3)}\in\F^A$ and representatives $c^{(h)}|_E\in\operatorname{Enc}_Z(y^{(h)})$, where $c^{(h)}\in C$ and $c^{(h)}|_A=y^{(h)}$. Then $c^{(1)}*c^{(2)}*c^{(3)}\in C^{*3}$, and by the construction of $b$,
\[
    \sum_{e\in E} b_e c^{(1)}_e c^{(2)}_e c^{(3)}_e
    =\sum_{a\in A} c^{(1)}_a c^{(2)}_a c^{(3)}_a
    =\sum_{a\in A} y^{(1)}_a y^{(2)}_a y^{(3)}_a.
\]
This is precisely the transversal-\(CCZ\) identity \eqref{eq:transversal-ccz-identity} for the logical
space \(\F^A\). By Lemma~\ref{lem:ccz-identity-implements-gate}, the corresponding physical
transversal gate implements logical \(CCZ\).\end{proof}

The following corollary is an immediate consequence of Theorem~\ref{thm:refined-GG}, and recovers Theorem~3.1 of \cite{GG24}.

\begin{corollary}[{\cite[Theorem~3.1]{GG24}}]\label{cor:GG-black-box}
Let $C \subseteq \F_q^n$ be an $[n,\ell]_q$ code such that $C$, $C^\perp$, and $C^{*3}$ have distances $d$, $d^\perp$, and $d^{(3)}$, respectively. For every $k<\min\{\ell,d,d^\perp,d^{(3)}\}$, 
 there is an \(
[[\,n-k,k,d_X\ge d-k,
        d_Z\ge d^\perp-k\,]]_{q}
\)
quantum CSS code that supports transversal $CCZ$.
In particular its  distance is at least
\(
        \min\{d,d^\perp\}-k .
\)
\end{corollary}

\begin{proof}
Let $\Omega = [n]$ be the index set of $C$. Choose a subset $A\subseteq\Omega$ of size $k<\ell$ such that $C|_A=\F^A$. The assumptions $k<d$ and $k<d^{(3)}$ imply $C\cap\F^A=\{0\}$ and $C^{*3}\cap\F^A=\{0\}$. 

For the distance terms, if \(c\in C\) and \(c|_A\ne 0\), then
\[
        |\supp(c)\setminus A|
        \ge \wt(c)-|A|
        \ge d-k.
\]
Thus \(d_A(C)\ge d-k\). Similarly,  \(d_A(C^\perp)\ge d^\perp-k\).
The result follows by Theorem~\ref{thm:refined-GG}, which gives the two 
bounds \(d_X\ge d-k\) and \(d_Z\ge d^\perp-k\). Taking the minimum gives the
 distance bound.
\end{proof}

\section{Instantiating the  puncturing theorem with algebraic expander codes}

\label{sec:aec-input}
Our plan is to instantiate Theorem~\ref{thm:refined-GG} with the algebraic expander code construction \cite{KT26}. In this section we prove some properties of these codes that will be used later  in
the quantum construction. 

Throughout this section, let \(m\geq 3, r\in (0,1)\) be fixed, and let  \(C_r=C(G,H;r,N)\subseteq\F^\Omega\) be the algebraic expander code from Theorem~\ref{prop:principal-aec-tanner} with length \(N\). In particular, \(C_r\) is the Tanner code
whose local constraints on the \(G\)- and \(H\)-orbits are RS codes of local dimension
\(\lceil r|O|\rceil\). This implies that the local constraints are of length at most \(|O|=O(N^{1/m})\).

\subsection{Schur powers}
\label{subsec:aec-schur-powers}

The first property we need is that  algebraic expander codes inherit the Schur-product
behavior of their local RS constraints.

\begin{lemma}[Schur-power containment]
\label{lem:schur-containment}
Let \(t\ge 1\) be fixed and suppose \(tr<1\). Then \(C_r^{*t}\subseteq C_{tr}\). Consequently, \(d(C_r^{*t})\ge d(C_{tr})=\Omega(N)\).
\end{lemma}
\begin{proof}
It is enough to check products of \(t\) codewords. Let
\(c^{(1)},\ldots,c^{(t)}\in C_r\). On every \(G\)-orbit \(O\), each \(c^{(i)}|_O\) is an evaluation of a polynomial of degree less than 
 \(k_O(r)=\lceil r|O|\rceil\). Hence the componentwise product
\[
        c^{(1)}|_O*\cdots*c^{(t)}|_O
\]
is an evaluation of a polynomial of degree at most \(t\bigl(k_O(r)-1\bigr)<t r|O|\le k_O(tr)\); in particular it lies in \(R_O(tr)\). The same argument applies on every
\(H\)-orbit. By Theorem~\ref{prop:principal-aec-tanner}
 \(c^{(1)}*\cdots*c^{(t)}\in C_{tr}\). Taking the span over all such products gives
\(C_r^{*t}\subseteq C_{tr}\). The distance bound follows from Theorem~\ref{prop:principal-aec-tanner}.
\end{proof}

For the construction of transversal \(CCZ\) codes we will use the case \(t=3\). Thus, whenever
\(r<1/3\), we have
\[
        C_r^{*3}\subseteq C_{3r}
        \qquad\text{and}\qquad
        d(C_r^{*3})=\Omega(N).
\]

\subsection{Exact dual distance and the limitation of the direct Golowich--Guruswami application}\label{subsec:aec-dual-distance-baseline}
We next compute the exact dual distance of the code \(C_r\). This result will be used as input to a black-box application of the Golowich--Guruswami framework \cite{GG24} (see Corollary~\ref{cor:GG-black-box}), and illustrates the  limitation of such a direct application. In particular, this application yields a CSS code whose dimension scales only as \(O(N^{1/m})\). Later, we build on this result with additional ideas to show that the dimension can in fact be made linear in the code length.

\begin{lemma}[Exact dual distance]
\label{lem:dual-distance}
Let \(k_r=\lceil rL\rceil\), where \(L=\min\{|G|,|H|\}\). Assume \(k_r<L\)
(which holds for all sufficiently large \(N\) for fixed \(r<1\)). Then
\(d(C_r^\perp)=k_r+1\).
\end{lemma}
\begin{proof}
Let \(\RS_\Omega(k_r)\) be the RS code obtained by evaluating polynomials of degree
less than \(k_r\) on \(\Omega\). We first show that \(\RS_\Omega(k_r)\subseteq C_r\). Indeed, for a nonnegative integer \(i\) if
\(0\le i<k_r\), then \(i<rL\). Since \(L\le |G|,|H|\), the \(g\)-base and \(h\)-base degrees of
\(X^i\) are both equal to \(i\). Hence \(X^i\) satisfies both local degree constraints defining
\(C_r\). Therefore every polynomial of degree less than \(k_r\) gives a codeword of \(C_r\).
This implies  that
\[
        C_r^\perp\subseteq \RS_\Omega(k_r)^\perp.
\]
Since \(\RS_\Omega(k_r)\) is an \([N,k_r]\) MDS code, its dual has minimum distance \(k_r+1\).
Thus \(d(C_r^\perp)\ge k_r+1\).

For the reverse inequality, let \(\Gamma\in\{G,H\}\) with \(|\Gamma|=L\). On every \(\Gamma\)-orbit, the
restriction of each codeword of \(C_r\) lies in a RS code of length \(L\) and dimension
\(k_r\). Hence every word in the dual of this local RS code, extended by zero outside
the orbit, belongs to \(C_r^\perp\). The dual of an \([L,k_r]\) RS code has a nonzero
word of weight \(k_r+1\). Therefore \(d(C_r^\perp)\le k_r+1\), and the claim follows.
\end{proof}

Combining the dual-distance computation with the Schur-power containment gives a direct black-box application of the Golowich--Guruswami framework \cite{GG24}.

\begin{corollary}
\label{cor:direct-instantiation}
Fix \(r<1/3\). Then the code \(C_r\) over \(\mathbb F_q\) gives a quantum CSS
code supporting transversal \(CCZ\) with parameters 
\[ [[\,N-O(N^{1/m}),\Omega(N^{1/m}),d_X=\Omega(N), d_Z=\Omega(N^{1/m})\,]]_{q}.
\]
\end{corollary}

\begin{proof}
By Theorem~\ref{prop:principal-aec-tanner}, \(C_r\) has linear dimension and linear distance. By
Lemma~\ref{lem:schur-containment}, \(C_r^{*3}\) has linear
distance. By Lemma~\ref{lem:dual-distance}, \(d(C_r^\perp)=\Theta(L)\). Applying
Corollary~\ref{cor:GG-black-box} with \(k\) chosen as a sufficiently small constant multiple of \(L\) gives a quantum
CSS code of length \(N-k=N-O(L)\), dimension \(k=\Omega(L)\), and distances \(d_X\geq d(C_r)-k,d_Z\geq d(C_r^\perp)-k\),
supporting transversal \(CCZ\). The result follows by noting that \(L=\Theta(N^{1/m})\).
\end{proof}

Thus, in the direct application of Corollary~\ref{cor:GG-black-box}, the dimension is limited by the dual distance
\(
    d(C_r^\perp)=\Theta(N^{1/m}).
\)
In the next subsection, we use an additional structural property of the code \(C_r\). This will allow us to construct a quantum CSS code with constant rate, improving on the vanishing rate  obtained by the direct black-box application above.

Finally, note that the parameters of the quantum CSS code in Corollary~\ref{cor:direct-instantiation} are worse than those of the quantum code constructed in~\cite{GG24}. Nevertheless, the present construction has one important advantage: it is possible to prove that the \(Z\)-stabilizers have low weight. We will prove this later for the improved construction that achieves constant rate.

\subsection{Preserving dual distance outside an interpolation set}\label{subsec:aec-robust-avoidance}
To improve the code dimension given in Corollary~\ref{cor:direct-instantiation}, we develop a robust puncturing argument for the nonzero codewords in \(C_r^\perp\). Specifically, we show the existence of a linear-sized set of coordinates that avoids the support of any nonzero dual codeword in a substantial number of positions. Consequently, after puncturing on this set, the minimum distance of the resulting punctured dual code remains a constant fraction of the minimum distance of \(C_r^\perp\).

The proof idea is the following. We show that if \(A\) is an interpolation set for a lower local-rate algebraic expander code, then no nonzero dual codeword of a higher local-rate code can be almost entirely supported on \(A\).

\begin{proposition}
\label{prop:dual-avoidance}
Let \(0<r_0<r<1\), and let \(A\subseteq\Omega\) be an interpolation set for \(C_{r_0}\), i.e.,  \(C_{r_0}|_A=\F^A\).  Then, for every nonzero \(z\in C_r^\perp\), \(|\supp(z)\setminus A|> \lfloor (r-r_0)L\rfloor\), hence 
\[
        d_A(C_r^\perp)>\lfloor (r-r_0)L\rfloor=\Omega(N^{1/m}) .
\]

\end{proposition}

\begin{proof}
We first prove an interpolation claim. Let \(s=\lfloor (r-r_0)L\rfloor\) and  \(T\subseteq\Omega\setminus A\) with \(|T|\le s\). We
claim that for every \(y\in\F^A\), there exists \(w\in C_r\) such that \(w|_A=y\) and \(w|_T=0\).

Let \(P_T(X)=\prod_{\beta\in T}(X-\beta)\) be the annihilator polynomial of the set \(T\). Since \(T\cap A=\emptyset\), we have \(P_T(\omega)\ne 0\) for every \(\omega\in A\). The assumption \(C_{r_0}|_A=\F^A\) implies that there exists  \(f\in W_{r_0,N}\) such that \(f(\omega)=y_{\omega}/P_T(\omega)\) for every \(\omega\in A\).
Let \(w\) be the evaluation vector of \(f\cdot P_T\) on \(\Omega\). Then, clearly \(w|_A=y\) and \(w|_T=0\).

It remains to show that \(w\) is a codeword of  \(C_r\). 
Note that although the global ordinary degree of \(f \cdot P_T\) may equal or exceed $N$, we will show that its \(g\)-base and \(h\)-base degrees guarantee that its evaluation vector $w$ satisfies all local Reed-Solomon constraints, and  therefore, $w \in C_r$ by the Tanner code characterization in Theorem~\ref{prop:principal-aec-tanner}.

Since \(|T|\le s<L\), the ordinary degree of \(P_T\) is
smaller than both \(|G|\) and \(|H|\). Hence
\(
        \deg_g(P_T)=\deg_h(P_T)=|T|.
\)
By the subadditivity of the base degree (Lemma~\ref{lem:base-degree-subadditivity}),
\[
        \deg_g(f\cdot P_T)< r_0|G|+|T|
        \le r_0|G|+(r-r_0)L
        \le r|G|,
\]
and similarly \(\deg_h(f\cdot P_T)<r|H|\). Therefore \(f\cdot P_T\) satisfies all local constraints defining
the  Tanner code \(C_r\), and hence \(w\in C_r\).

Now let \( z\in C_r^\perp\) be a  dual codeword, and set \(T=\supp(z)\setminus A\). Suppose that \(|T|\le s\). By the interpolation claim, for every \(y\in\F^A\) there is
\(w\in C_r\) such that \(w|_A=y\) and \(w|_T=0\). Therefore, 
\[
        0=\langle z,w\rangle=\sum_{a\in A} z_a y_a+\sum_{e\in \Omega\setminus A}z_e w_e=\sum_{a\in A} z_a y_a,
\]
where the second sum vanishes because \(\supp(z)\setminus A=T\) and \(w|_T=0\).
Since this holds for any \(y\in\F^A\) it follows that \(z|_A=0\). Therefore \(\supp(z)\subseteq T\), so
\(\wt(z)\le s<\lceil rL\rceil+1=d(C_r^\perp)\) by Lemma~\ref{lem:dual-distance}. Hence, necessarily $z=0$. This implies that if \(z\) is nonzero, then \(|\supp(z)\setminus A|\ge s+1\).
\end{proof}
\begin{remark}
\label{prop:explicit-A}
The interpolation set \(A\) in Proposition~\ref{prop:dual-avoidance} can be constructed deterministically in polynomial time by greedily selecting coordinates corresponding to linearly independent columns of a generator matrix of \(C_{r_0}\). This ensures that \(C_{r_0}|_A = \mathbb{F}^A\).
\end{remark}

\begin{lemma}
\label{lem:dual-avoidance-upper}
Let \(0<r_0<r<1/3\), and let \(A\subseteq\Omega\) be a nonempty interpolation
set for \(C_{r_0}\). Then, for all sufficiently large \(N\),
\[
        d_A(C_r^\perp)
           =
        \Theta(N^{1/m}).
\]

\end{lemma}

\begin{proof}
By  Proposition~\ref{prop:dual-avoidance} we only need to show that 
\(d_A(C_r^\perp)=O(N^{1/m})\).
Choose \(a\in A\), and let \(O\in\mathcal O\) be one of the two local orbits
containing \(a\). Since \(A\) is an interpolation set for \(C_{r_0}\), the set
\(A\cap O\) has size at most \(k_O(r_0)\). Indeed, every \(c\in C_{r_0}\) satisfies \(c|_O\in R_O(r_0)\), so
\(C_{r_0}|_{A\cap O}\subseteq R_O(r_0)|_{A\cap O}\), and the latter space has dimension at most \(\dim R_O(r_0)=k_O(r_0)\).
On the other hand \(C_{r_0}|_{A\cap O}=\F^{A\cap O}\), because \(A\) is an interpolation set for \(C_{r_0}\).
Hence \(|A\cap O|\le k_O(r_0)\).

Since \(r_0+r<1\), for all sufficiently large \(N\) we have
\[
        k_O(r_0)+k_O(r)\le |O|.
\]
Thus \(O\setminus A\) contains a subset \(T\) of size \(k_O(r)\). The local code
\(R_O(r)\) is MDS, so its dual contains a nonzero word supported on
\(\{a\}\cup T\), with nonzero coordinate at \(a\). Extending this word by zero
outside \(O\) gives an element of \(C_r^\perp\). It is nonzero on \(A\), and its
support outside \(A\) is contained in \(T\). Hence
\[
        d_A(C_r^\perp)\le |T|=k_O(r).
\]
Since every local orbit has size \(O(N^{1/m})\), the claimed upper bound follows.
\end{proof}

We now present and prove our main construction of quantum codes over a growing alphabet.
\begin{theorem}[Main growing-alphabet construction]
\label{thm:main}
For any fixed \(m\geq 3\) and prime \(p\), there exists an explicit family of  quantum CSS codes with parameters
\[
    [[\,\Theta(N),\Theta(N),d_X=\Theta(N), d_Z=\Theta(N^{1/m})\,]]_{q},
\]
where \(N\) is the length of the underlying algebraic expander code. The codes support transversal \(CCZ\) gates, and their \(Z\)-stabilizer generators have weight  \(O(N^{1/m})\), where  \(q\) is
  a power of \(p\) and grows polynomially with 
\(N\).

\end{theorem}

\begin{proof}
Fix constants \(0<r_0<r<1/3\), and let \(C_r=C(G,H;r,N)\) be the code from Theorem~\ref{prop:principal-aec-tanner}. Set \(L=\min\{|G|,|H|\}\). Let \(R_0\) be the rate of \(C_{r_0}\), and let \(\delta_r,\delta_{3r}>0\) denote the relative minimum distances of \(C_r\) and \(C_{3r}\), respectively.

Choose a constant \(0<\alpha<\min\{R_0,\delta_r,\delta_{3r}\}\). Since \(\alpha<R_0\), we have \(\lfloor \alpha N\rfloor\leq \dim C_{r_0}\). By Remark~\ref{prop:explicit-A}, applied to \(C_{r_0}\), there exists an efficiently constructible set \(A\subseteq\Omega\) of size \(\lfloor \alpha N\rfloor\) such that \(C_{r_0}|_A=\F^A\). Since \(C_{r_0}\subseteq C_r\), it follows that \(C_r|_A=\F^A\).

Moreover, since \(|A|=\lfloor\alpha N\rfloor<\delta_r N=d(C_r)\), we have \(C_r\cap \F^A=\{0\}\). Similarly, Lemma~\ref{lem:schur-containment} gives \(C_r^{*3}\subseteq C_{3r}\), and since \(|A|<\delta_{3r}N=d(C_{3r})\), we also have \(C_r^{*3}\cap \F^A=\{0\}\).

Finally, 
by Lemma~\ref{lem:dual-avoidance-upper},
\(
        d_A(C_r^\perp)=
        \Theta(N^{1/m}).
\)
Moreover,
\[
        d_A(C_r)
        =
        \min\{|\supp(c)\setminus A|:c\in C_r,\ c|_A\neq 0\}
        \ge d(C_r)-|A|
        \ge \delta_r N-|A|=\Omega(N),
\]
since \(\alpha<\delta_r\).

Applying Theorem~\ref{thm:refined-GG} with \(C=C_r\), we obtain a CSS code of
length
\(
        N-|A|=\Theta(N)
\),
dimension \(|A|=\Theta(N)\), and transversal \(CCZ\). The two CSS distances then  satisfy
\[
        d_X
        =
        d_A(C_r)=\Theta(N),\qquad d_Z=d_A(C_r^\perp)=
        \Theta(N^{1/m}).
\]
It remains only to prove the claim on the weight of the \(Z\)-stabilizer generators. This will be established in Section~\ref{sec:one-sided-locality}, specifically in Proposition~\ref{prop:Z-low-weight}.
\end{proof}

\paragraph{Asymmetric interpretation and a comparison with a direct sum of quantum RS codes.}
The construction is asymmetric both in its logical distances and in the weights
of its stabilizer generators. The larger distance is the \(X\)-distance,
\[
        d_X=\Theta(N),
\]
whereas
\[
        d_Z=\Theta(N^{1/m}).
\]
At the same time, the low-weight generators we obtain are \(Z\)-stabilizer
generators, which are the checks used to detect \(X\)-type errors. Thus the two
asymmetric features are aligned: the code has stronger protection against
\(X\)-type errors, and the syndrome checks relevant for those errors have
sublinear weight. This makes the construction most naturally suited to biased
noise models in which \(X\)-type errors are the dominant error mechanism.

This distinction is hidden if one considers only the usual quantum distance
\(
        d=\min\{d_X,d_Z\}.
\)
At that level, the construction has length \(\Theta(N)\), dimension
\(\Theta(N)\), distance \(\Theta(N^{1/m})\), transversal \(CCZ\), and
\(Z\)-stabilizer generators of weight \(O(N^{1/m})\). These  parameters
can also be obtained by taking \(\Theta(N^{1-1/m})\) disjoint copies of a
quantum RS code of length \(\Theta(N^{1/m})\), constant rate, and
linear \(X\)- and \(Z\)-distances, with each block supporting transversal
\(CCZ\). The resulting direct-sum code has total length \(\Theta(N)\),
dimension \(\Theta(N)\), and
\[
        d_X=d_Z=\Theta(N^{1/m}).
\]
Moreover, both its \(X\)- and \(Z\)-stabilizer spaces admit generating sets of
weight \(O(N^{1/m})\).

The distinction of our construction is therefore not its usual minimum
distance or its stabilizer locality alone, but its asymmetric distance profile:
one CSS distance is linear in the total block length, while the other is
sublinear. The construction is most naturally viewed as a family of
asymmetric CSS codes with transversal \(CCZ\) and one-sided sublinear
stabilizer locality.

\subsection{Sublinear \texorpdfstring{$Z$}{Z}-stabilizer generators}
\label{sec:one-sided-locality}
In this section we show that the quantum CSS code constructed in Theorem~\ref{thm:main} has low-weight \(Z\)-stabilizers, i.e., its \(Z\)-stabilizers can be generated by low-weight checks.

This property does not follow automatically from the fact that the algebraic expander code \(C_r\) is a Tanner code and that \(Q_Z^\perp=(C_r^\perp\cap \F^E)|_E\), where \(E=\Omega\setminus A\). Indeed, \(Q_Z^\perp\) is obtained by shortening \(C_r^\perp\) at \(A\), and shortening a locally generated code does not, in general, preserve local generation.

To overcome this issue, we rely on a structural property of the set \(A\), namely the interpolation property \(C_{r_0}|_A=\F^A\) (see Proposition~\ref{prop:dual-avoidance}). This property forces \(A\) to be sparse within every local orbit. Exploiting this sparsity, we show that \(Q_Z^\perp\) admits a generating set consisting of low-weight checks.

Recall that for \(0<r<1\), \(C_r=C(G,H;r,N)\), and \(\mathcal O\) denotes the collection of all
\(G\)- and \(H\)-orbits in \(\Omega\).
Moreover, for \(O\in\mathcal O\), \(R_O(r)\subseteq\F^O\) denotes the local RS code appearing in the Tanner definition of the code 
\(C_r\) (see Theorem~\ref{prop:principal-aec-tanner}), and write \(k_O(r)=\dim R_O(r)\). We view \(R_O(r)^\perp\) as a subspace of
\(\F^\Omega\) by extending vectors by zero outside \(O\).

\begin{proposition}[Low-weight \(Z\)-stabilizers]
\label{prop:Z-low-weight}
For the quantum CSS code given in Theorem~\ref{thm:main}, the \(Z\)-stabilizer
space \(Q_Z^\perp\) admits a generating set consisting of checks of weight at
most \(|G|+|H|=O(N^{1/m})\).
\end{proposition}
More explicitly, the generating set consists of two types of checks. The first type is given by a basis for each dual local code whose support is entirely contained in \(E\); these checks are supported on a single orbit. The second type consists of ``gluing'' checks \(\eta_a\) for each \(a\in A\), where each such check is supported on the union of one \(G\)-orbit and one \(H\)-orbit, and therefore its weight is at most   \(|G|+|H|\).

\begin{proof}
Recall that \(A\subseteq\Omega\) satisfies \(C_{r_0}|_A=\F^A\), where \(0<r_0<r<1/3\). Furthermore, the CSS code is
\[
        Q_X=C_r^\perp|_E,
        \qquad
        Q_Z=C_r|_E, \qquad
E=\Omega\setminus A.
\]
Then, by puncturing-shortening duality \eqref{eq:puncture-shorten-duality},
\(
        Q_Z^\perp
        =
        (C_r|_E)^\perp
        =
        (C_r^\perp\cap \F^E)|_E .
\)
Thus, it is enough to construct low-weight generators for \(C_r^\perp\cap \F^E\).

For every orbit \(O\in\mathcal O\), the set \(A\cap O\) has size at most \(k_O(r_0)\). Indeed, \(C_{r_0}|_{A\cap O}\subseteq R_O(r_0)|_{A\cap O}\) has dimension at most \(k_O(r_0)\), while \(C_{r_0}|_{A\cap O}=\F^{A\cap O}\) since \(A\) is an interpolation set for \(C_{r_0}\); hence \(|A\cap O|\le k_O(r_0)\).
Next, since \(r_0 < r < 1/3\), we have $r_0+r < 2/3$. For sufficiently large \(N\) (and hence large orbit size \(|O| = \Theta(N^{1/m})\)), it holds that \(k_O(r_0)+k_O(r) = \lceil r_0|O| \rceil + \lceil r|O| \rceil \leq |O|\).
Therefore, 
\[
      |O\setminus A|=|O|-  |A\cap O|\ge |O|- k_O(r_0)\ge k_O(r).
\]
Hence \(O\setminus A\) contains an information set of size \(k_O(r)\) for the local MDS code \(R_O(r)\).

Fix such an information set \(I_O\subseteq O\setminus A, |I_O|=k_O(r)\), for every orbit \(O\in\mathcal O\). Since \(R_O(r)\) is MDS, every coordinate in \(O\setminus I_O\) is a linear combination of the coordinates in \(I_O\). Equivalently, the dual space \(R_O(r)^\perp\) has a basis
\[
        \{h_{O,b}: b\in O\setminus I_O\},
\]
where \(h_{O,b}\) is supported on \(I_O\cup\{b\}\) and is normalized so that \((h_{O,b})_b=1\). We view these rows as vectors in \(\F^\Omega\), extended by zero outside \(O\).

Since \(I_O\subseteq E\), the rows \(h_{O,b}\) with \(b\in E\) are supported entirely in \(E\). The rows with \(b\in A\cap O\) are the only rows from the orbit \(O\) that have a nonzero coordinate in \(A\). In particular, for each \(a\in A\), there are exactly two such rows touching the coordinate \(a\): one from the \(G\)-orbit through \(a\), i.e., \(h_{G\cdot a,a}\), and one from the \(H\)-orbit through \(a\), i.e., \(h_{H\cdot a, a}\). 
Both have \(a\)-coordinate equal to \(1\), and neither has any other nonzero coordinate in \(A\).
Therefore
\(
        \eta_a:=h_{G\cdot a,a}-h_{H\cdot a,a}
\)
belongs to \(C_r^\perp\cap \F^E\). Moreover,
\(
        \supp(\eta_a)
\) is contained in the union of the orbits \(G\cdot a\) and \(H\cdot a\), 
so \(\wt(\eta_a)\le |G|+|H|\).

We now prove that these two sets of rows supported inside \(E\), i.e.,
\begin{equation}
\label{eq:generating-set}
        \{h_{O,b}: O \in \mathcal{O},  b\in (O\setminus I_O)\cap E\}\text { and } \{\eta_a:a\in A\},
\end{equation}
generate \(C_r^\perp\cap \F^E\). Let \(z\in C_r^\perp\cap \F^E\). Since \(C_r\) is a Tanner
code, equivalently \(C_r^\perp\) is generated by the local dual spaces \(R_O(r)^\perp\), \(z\) can be written as 
\[
        z=\sum_{O\in\mathcal O}\sum_{b\in O\setminus I_O}
        \lambda_{O,b} h_{O,b}=\sum_{O\in\mathcal O}\sum_{b\in (O\setminus I_O)\cap E}
        \lambda_{O,b} h_{O,b}+ \sum_{O\in\mathcal O}\sum_{a\in (O\setminus I_O)\cap A}
        \lambda_{O,a} h_{O,a},
\]
for some scalars \(\lambda_{O,b}\). Each summand in the first double sum is in the first set \(\{h_{O,b}: O \in \mathcal{O},  b\in (O\setminus I_O)\cap E\}\). We show next that the second double sum is spanned by the \(\eta_a\)'s. For a fixed \(a\in A\), the only two rows in the above expansion
that can contribute to the \(a\)-coordinate are \(h_{G\cdot a,a}\) and \(h_{H\cdot a,a}\). Since \(z\) is supported on \(E\), its \(a\)-coordinate is zero. Hence
\(
        \lambda_{G\cdot a,a}+\lambda_{H\cdot a,a}=0
\) since the \(a\)-coordinate of both \(h_{G\cdot a,a}\) and \(h_{H\cdot a,a}\) is \(1\). Therefore,
\[
        \lambda_{G\cdot a,a}h_{G\cdot a,a}+\lambda_{H\cdot a,a}h_{H\cdot a,a}
        =
        \lambda_{G\cdot a,a}(h_{G\cdot a,a}-h_{H\cdot a,a})
        =
        \lambda_{G\cdot a,a}\eta_a .
\]
Thus every contribution involving a coordinate of \(A\) can be replaced by a multiple of the
corresponding row \(\eta_a\), and hence the second double sum is spanned by the \(\eta_a\)'s.
This proves that \(C_r^\perp\cap\F^E\) is generated by the elements in \eqref{eq:generating-set}.
Restricting these generators to \(E\) gives a generating set for \(Q_Z^\perp\).

Each local row \(h_{O,b}\) is supported on a single orbit, and hence has weight at most
\(\max\{|G|,|H|\}\). Each row \(\eta_a\) is supported on the union of one \(G\)-orbit and one
\(H\)-orbit, and hence has weight at most \(|G|+|H|\), and the result follows. 
\end{proof}

\begin{remark}[\(X\)-stabilizers weight]
In contrast to the \(Z\)-stabilizers, the situation for the \(X\)-stabilizers is much worse. Indeed, the \(X\)-stabilizer space is \(Q_X^\perp=(C_r\cap\F^E)|_E\).

Every nonzero element of this space arises from a nonzero codeword of \(C_r\) that vanishes on \(A\), and therefore has weight at least \(d(C_r)=\Omega(N)\), i.e., linear in the code length.
\end{remark}

\section{Alphabet reduction}
\label{sec:one-shot-alphabet-reduction}

A caveat of the quantum CSS codes constructed in Theorem~\ref{thm:main} is
that they are defined over an alphabet whose size grows polynomially  with the code
length.

In this section, we address this issue by reducing the alphabet directly to the
prime field%
\footnote{One can similarly reduce the alphabet to any intermediate subfield.
For clarity, we only treat the reduction to the prime field.}.
We do so by concatenating with multiplication-friendly inner quantum CSS codes and applying
the Golowich--Guruswami alphabet-reduction operation
\cite[Proposition~5.1]{GG24} as a black box.

Beyond this black-box application, we also need to track the weights of the
\(Z\)-stabilizer generators.  Specifically, we show that the low-weight
\(Z\)-stabilizer generating set established in
Proposition~\ref{prop:Z-low-weight} survives the alphabet reduction, with each
generator weight increasing by at most a factor equal to the length of the
inner multiplication-friendly quantum CSS code.

Finally, for the inner codes, we use a projective-multiplicity generalization
of the multiplication-friendly codes employed in \cite{GG24,GJX17}.  Unlike the
multiplicity-one affine construction, this variant works over arbitrary field
extensions and in arbitrary characteristic, including characteristic \(2\). 

\subsection{Multiplication-friendly inner codes}Let $\mathbb{K} = \mathbb{F}_{q^\kappa}$. 
A collection of classical codes \((C^{(h)};\enc^{(h)})_{h=1}^a\), with
\(C^{(h)}=\operatorname{im}(\enc^{(h)})\subseteq \mathbb F_q^\nu\), is
\(a\)-multiplication-friendly for $\mathbb{K}/\mathbb{F}_q$ if each $\enc^{(h)}: \mathbb{K} \to \mathbb{F}_q^\nu$ is an $\mathbb{F}_q$-linear injective map, and there exists an $\mathbb{F}_q$-linear map $\dec: \mathbb{F}_q^\nu \to \mathbb{K}$ such that$$ \dec(\enc^{(1)}(x_1) * \cdots * \enc^{(a)}(x_a)) = x_1 \cdots x_a $$for all $x_1, \ldots, x_a \in \mathbb{K}$. 

We also use the corresponding quantum notion. A collection
\(
        \bigl(Q_{\rm in}^{(h)}
        =
        \operatorname{CSS}(Q_X^{(h)},Q_Z^{(h)};\Enc_Z^{(h)})\bigr)_{h=1}^a
\)
of CSS codes over \(\F_q\), 
with \(Q_X^{(h)},Q_Z^{(h)}\subseteq\F_q^\nu\) and \(\Enc_Z^{(h)}:\mathbb K\xrightarrow{\sim} Q_Z^{(h)}/{Q^{(h)}_X}^\perp\), is
\(a\)-multiplication-friendly for \(\K/\F_q\) if there exists an
\(\F_q\)-linear map
\(
        \Dec_Z:\F_q^\nu\to\K
\)
such that, for all \(x_1,\ldots,x_a\in\K\) and all representatives
\(z^{(h)}\in \Enc_Z^{(h)}(x_h)\), it holds that 
\[
        \Dec_Z\bigl(z^{(1)}*\cdots*z^{(a)}\bigr)
        =
        x_1\cdots x_a .
\]

A collection of classical multiplication-friendly codes $(C^{(h)}; \enc^{(h)})_{h=1}^a$ gives a quantum multiplication-friendly code via the trivial CSS
conversion
\[
        Q_{\rm in}^{(h)}
        =
        \operatorname{CSS}(\F_q^\nu,C^{(h)};\enc^{(h)}),
\]
for each \(h\in[a]\).
Indeed, since
\((\F_q^\nu)^\perp=\{0\}\), the \(Z\)-encoding cosets are single codewords,
and the quantum multiplication-friendly identity is exactly the classical one.
Each such  CSS code has parameters
\(
        [[\,\nu,\kappa,\ge 1\,]]_q
\) 
when \(\K=\F_{q^\kappa}\), and its \(Z\)-stabilizer space is
\((C^{(h)})^\perp\), which has a generating matrix with row weight at most
\(\nu\).

%Furthermore, every $\mathbb{F}_q$-linear functional on $\mathbb{K}$ has a physical $Z$-logical representative of weight at most $\nu$.

For our transversal $CCZ$ application, we use arity $a=4$. The first three encoders are applied to the three quantum blocks, while the fourth encoder is used to expand the coefficient vector in the transversal $CCZ$ identity.

To achieve the desired parameters in the alphabet reduction, we use a projective-multiplicity generalization of the RS multiplication-friendly construction used in
\cite{GJX17,GG24}.  The generalization is a univariate multiplicity
interpolation construction over \(\mathbb P^1(\mathbb F_q)\): it allows
derivative evaluations with prescribed local multiplicities and incorporates
the point at infinity.  The full  construction and proof are deferred to
Appendix~\ref{app:projective-multiplicity-mfc}.

\begin{proposition}[Projective-multiplicity multiplication-friendly codes]\label{prop:projective-multiplicity-mfc}Let $q$ be a prime power, and let $a, \kappa \ge 1$. Set $D = a(\kappa - 1)$. Let $(\ell_P)_{P \in \mathbb{P}^1(\mathbb{F}_q)}$ be nonnegative integers satisfying $\sum_{P \in \mathbb{P}^1(\mathbb{F}_q)} \ell_P = D + 1$. Then there is an $a$-multiplication-friendly collection of classical $[\nu, \kappa]_q$ codes for $\mathbb{F}_{q^\kappa}/\mathbb{F}_q$, where$$ \nu = \sum_{P \in \mathbb{P}^1(\mathbb{F}_q)} N_{a,\kappa}(\ell_P), \qquad N_{a,\kappa}(\ell) = \#\{(i_1, \ldots, i_a) \in \{0, \ldots, \kappa - 1\}^a : i_1 + \cdots + i_a < \ell\}. $$\end{proposition}\begin{proof}See Appendix~\ref{app:projective-multiplicity-mfc}.\end{proof}

\begin{remark}[Choice of multiplicities]
\label{rem:balanced-multiplicities}
The construction requires a multiplicity profile
\((\ell_P)_{P\in\mathbb P^1(\mathbb F_q)}\) satisfying
\(
\sum_{P\in\mathbb P^1(\mathbb F_q)} \ell_P = a(\kappa-1)+1 .
\)
For such a profile, the length of the inner code is
\[
\nu=
\sum_{P\in\mathbb P^1(\mathbb F_q)} N_{a,\kappa}(\ell_P)
=
\sum_{P\in\mathbb P^1(\mathbb F_q)}\#\{(i_1,\ldots,i_a)\in\{0,\ldots,\kappa-1\}^a:
i_1+\cdots+i_a<\ell_P\}.
\]
In particular, \(N_{a,\kappa}(\ell)\le \kappa^a\) for every \(\ell\), and
therefore every choice of multiplicities gives \(\nu\le (q+1)\kappa^a\).

A convenient choice is to distribute the multiplicities as evenly as possible.
Write \(a(\kappa-1)+1=u(q+1)+b\), where \(0\le b<q+1\). Assign multiplicity
\(u+1\) to \(b\) projective points and multiplicity \(u\) to the remaining
\(q+1-b\) points. This choice gives \(\nu=O(q\kappa^a)\). Moreover, when \(q\)
and \(a\) are fixed, it gives \(\nu=\Theta(\kappa^a)\): indeed, at least one
projective point has multiplicity \(\Theta(\kappa)\), and then
\(N_{a,\kappa}(\ell)=\Theta(\kappa^a)\) for \(\ell=\Theta(\kappa)\).
In particular, for the \(CCZ\) reduction below we take \(a=4\) and fixed
\(q=p\), obtaining \(\nu=\Theta(\kappa^4)\).

When \(a(\kappa-1)+1\le q+1\), one may assign multiplicity \(1\) to
\(a(\kappa-1)+1\) projective points and multiplicity \(0\) to the remaining
points. This recovers the standard RS multiplication-friendly
construction, with the minor modification that the point at infinity may also
be used as an evaluation point.

The multiplicity variant becomes useful in the fixed-\(q\) regime, where
\(\kappa\) grows and ordinary RS interpolation over
\(\mathbb F_q\) no longer provides sufficiently many evaluation points.
Moreover, the projective formulation is sometimes necessary even in very small
parameter regimes. For example, when \(q=a=\kappa=2\), ordinary 
RS interpolation is impossible, and the point at infinity must be
used in order to obtain the required multiplication-friendly property.
\end{remark}

\subsection{Alphabet reduction with low-weight \texorpdfstring{$Z$}{Z}-stabilizers}
\label{subsec:alphabet-reduction-z-locality}

We now apply CSS code concatenation together with the restriction operation
from Lemma~\ref{lem:css-restriction}. Let \(\K=\F_{p^\kappa}\), and view
\(\K\) as a \(\kappa\)-dimensional vector space over \(\F_p\). Let
\[
        Q_{\rm in}
        =
        \operatorname{CSS}
        (Q_{X,{\rm in}},Q_{Z,{\rm in}};
        \Enc_{X,{\rm in}},\Enc_{Z,{\rm in}})
\]
be an \(\F_p\)-ary CSS code with logical space \(\K\) and parameters
\(
        [[\,\nu,\kappa,d_{X, \rm in}, d_{Z, \rm in}\,]]_p .
\)
Let
\[
        Q_{\rm out}
        =
        \operatorname{CSS}
        (Q_{X,{\rm out}},Q_{Z,{\rm out}};
        \Enc_{X,{\rm out}},\Enc_{Z,{\rm out}})
\]
be a \(\K\)-ary CSS code with logical space \(\K^{k_{\rm out}}\) and parameters
\(
        [[\,n_{\rm out},k_{\rm out},d_{X, \rm out}, d_{Z, \rm out}\,]]_{\K}.
\)

We choose the inner \(X/Z\)-encoding maps to be compatible with respect to the
trace bilinear form
\(
        (x,z)\longmapsto \operatorname{Tr}_{\K/\F_p}(xz),
\)
and the outer \(X/Z\)-encoding maps to be compatible with respect to the
standard \(\K\)-bilinear form on \(\K^{k_{\rm out}}\). With these choices, the
concatenated code
\(
        Q_{\rm in}\circ Q_{\rm out}
\)
is the \(\F_p\)-ary CSS code whose encoding maps are
\[
        \Enc_Z
        =
        (\Enc_{Z,{\rm in}})^{\oplus n_{\rm out}}
        \circ
        \Enc_{Z,{\rm out}},
        \qquad
        \Enc_X
        =
        (\Enc_{X,{\rm in}})^{\oplus n_{\rm out}}
        \circ
        \Enc_{X,{\rm out}},
\] 
and it has parameters
\[
        \left[\left[
        \nu n_{\rm out},\
        \kappa k_{\rm out},\
        d_X\ge d_{X,\rm in}d_{X,\rm out},\
        d_Z\ge d_{Z,\rm in}d_{Z,\rm out}
        \right]\right]_p .
\]
The lower bound on the distances follows, since any nontrivial logical \(X\)-operator of the concatenated code must act
nontrivially on at least \(d_{X,\rm out}\) inner blocks, and its restriction to
each such block represents a nontrivial logical \(X\)-operator of the inner
code. Hence its weight is at least
\(d_{X,\rm in}d_{X,\rm out}\). The same argument applies to \(Z\)-operators.

We use the alphabet-reduction theorem of 
\cite[Proposition~5.1]{GG24}, stated next in our notation and where we keep track of two  distances \(d_X,d_Z\).
\begin{proposition}[{\cite[Proposition~5.1]{GG24}}]
\label{prop:GG-restricted-concatenation}
Let \(\K=\F_{p^\kappa}\). Let
\((Q_{\mathrm{in}}^{(h)})_{h=1}^4\) be a 4-multiplication-friendly collection of inner CSS codes
over \(\F_p\), where \(Q_{\mathrm{in}}^{(h)}\) has parameters
\([[\,\nu,\kappa,d_{X, \mathrm{in}}^{(h)}, d_{Z, \mathrm{in}}^{(h)}\,]]_p\) and logical space \(\K\). Let
\((Q_{\mathrm{out}}^{(h)})_{h=1}^3\) be a triple of outer CSS codes over \(\K\), where
\(Q_{\mathrm{out}}^{(h)}\) has parameters
\([[\,n_{\mathrm{out}},k_{\mathrm{out}},d_{X, \mathrm{out}}^{(h)}, d_{Z, \mathrm{out}}^{(h)}\,]]_{\K}\), and suppose this
triple supports transversal \(CCZ_{\K}\).

Choose \(X\)-encoding maps for the inner and outer CSS codes as follows. For
\(h=1,\ldots,4\), choose \(\Enc_{X,{\rm in}}^{(h)}\) compatible with
\(\Enc_{Z,{\rm in}}^{(h)}\) with respect to the bilinear mapping 
\[
        (x,z)\longmapsto \operatorname{Tr}_{\K/\F_p}(xz).
\]
For \(h=1,2,3\), choose \(\Enc_{X,{\rm out}}^{(h)}\) compatible with
\(\Enc_{Z,{\rm out}}^{(h)}\) with respect to the standard \(\K\)-bilinear form
on \(\K^{k_{\rm out}}\). Such choices exist by
Lemma~\ref{lem:compatible-X-encoding}. The concatenation below is taken with
these compatible \(X/Z\)-encodings.

Fix an integer \(1\le \tau\le p\) with \(3(\tau-1)<\kappa\). Choose a subset
\(T\subseteq\F_p\) with \(|T|=\tau\), and identify \(S_\tau:=\F_p[X]^{<\tau}\) with
\(\F_p^T\) by evaluation on \(T\). Also fix an identification
\(\K\simeq \F_p[X]/(\gamma(X))\simeq \F_p[X]^{<\kappa}\), where
\(\gamma(X)\in\F_p[X]\) is irreducible of degree \(\kappa\). Thus \(S_\tau\subseteq\K\).

For \(h=1,2,3\), define the restricted concatenated code by
\[
        \widetilde Q^{(h)}
        :=
        (Q_{\mathrm{in}}^{(h)}\circ Q_{\mathrm{out}}^{(h)})\big|_{S_\tau^{k_{\mathrm{out}}}} .
\]
Then \(\widetilde Q^{(h)}\) has parameters
\[
        [[\,\nu n_{\mathrm{out}},\ \tau k_{\mathrm{out}},\ d_X
        \ge
d_{X,\mathrm{in}}^{(h)}d_{X,\mathrm{out}}^{(h)},\
        d_Z
        \ge
d_{Z,\mathrm{in}}^{(h)}d_{Z,\mathrm{out}}^{(h)}\,]]_p .
\]
Moreover, the triple
\((\widetilde Q^{(1)},\widetilde Q^{(2)},\widetilde Q^{(3)})\)
supports transversal \(CCZ_p\).
\end{proposition}
The bound on the distances follows from the bound on the distances of the concatenated code and by  Lemma~\ref{lem:css-restriction}.

We now combine Proposition~\ref{prop:GG-restricted-concatenation} with the quantum CSS code from
Theorem~\ref{thm:main} and the projective-multiplicity multiplication-friendly codes from
Proposition~\ref{prop:projective-multiplicity-mfc}.

\begin{theorem}[Fixed-prime-field alphabet reduction]
\label{thm:one-shot-alphabet-reduction}
For any fixed \(m\ge 3\) and prime \(p\), there exists an explicit family of
\(\F_p\)-ary quantum CSS code triples
\(
        (Q^{(1)},Q^{(2)},Q^{(3)})
\)
supporting transversal \(CCZ_p\) gates, where each code has parameters 
\[
        \left[\left[
        n,\ 
        \Theta\!\left(\frac{n}{(\log n)^4}\right),\ 
         \   d_X
        =
        \Omega\!\left(\frac{n}{(\log n)^4}\right),\ d_Z
        =
        \Omega\!\left(\frac{n^{1/m}}{(\log n)^{4/m}}\right)\!
        \right]\right]_p.
\]
Moreover, the \(Z\)-stabilizer generators have weight at most
\(
      O\!\left(n^{1/m}(\log n)^{1-4/m}\right).
\)
\end{theorem}

\begin{proof}
We apply Proposition~\ref{prop:GG-restricted-concatenation} using the quantum CSS code from Theorem~\ref{thm:main} as the outer code, and the quantum multiplication-friendly codes  as inner codes, derived from the code in Proposition~\ref{prop:projective-multiplicity-mfc}.

By Theorem~\ref{thm:main}, for every fixed \(m\ge 3\) and characteristic \(p\), there is an
explicit family of  CSS codes over fields \(\K=\F_{p^\kappa}\) with parameters
\[
    [[\,\Theta(N),\Theta(N),d_X=\Theta(N), d_Z=\Theta(N^{1/m})\,]]_{p^\kappa},
\]
where \(N\) is the length of the underlying algebraic expander code.
These codes support transversal \(CCZ_{\K}\), and their \(Z\)-stabilizer generators have 
weight \(O(N^{1/m})\). Since the algebraic expander code construction evaluates on a set
\(\Omega\subseteq\K\) of size \(N\), we have \(|\K|\ge N\). On the other hand,
\(|\K|\) is polynomial in \(N\). Thus, for fixed characteristic \(p\),
\(
        \kappa=\Theta(\log N).
\)

Next, by Proposition~\ref{prop:projective-multiplicity-mfc} with arity \(a=4\), using the
balanced multiplicity profile from Remark~\ref{rem:balanced-multiplicities}, we get a
classical 4-multiplication-friendly collection for \(\K/\F_p\) of length
\(\nu=\Theta(\kappa^4)=O(\log^4 N)\).

We convert these classical codes to four inner quantum  codes by the
trivial construction
\(Q_{\mathrm{in}}^{(h)}=\operatorname{CSS}(\F_p^\nu,C_{\mathrm{in}}^{(h)};\enc^{(h)})\) for \(h=1,\ldots,4\), each with parameters \([[\,\nu,\kappa,d_{X, \mathrm{in}}, d_{Z, \mathrm{in}}\ge 1\,]]_p\). Their \(Z\)-stabilizers have
 weight at most \(\nu\). Moreover, every \(\F_p\)-linear functional on \(\K\) has a physical inner
\(Z\)-logical representative of weight at most \(\nu\).

Choose any fixed
\(1\le \tau\le p\) and let \(\kappa\) be sufficiently large such that
\(3(\tau-1)<\kappa\).
By Proposition~\ref{prop:GG-restricted-concatenation}, concatenating the
inner and outer CSS codes and then restricting the logical alphabet yields an
explicit triple of \(\F_p\)-ary CSS codes supporting transversal
\(CCZ_p\).

The block length is
\(
        n=\Theta(\nu N),
\)
and the dimension is
\(
        \tau\Theta(N)
        =
        \Theta(N)
        =
        \Theta\!\left(\frac{n}{(\log n)^4}\right).
\)
The outer code from Theorem~\ref{thm:main} has
\(
        d_X=\Theta(N),        d_Z=\Theta(N^{1/m})
\), 
while each inner code has both CSS distances at least \(1\). Hence
Proposition~\ref{prop:GG-restricted-concatenation} gives
\[
        d_X
        =
        \Omega(N)
        =
        \Omega\!\left(\frac{n}{(\log n)^4}\right),
\qquad        d_Z
        =
        \Omega(N^{1/m})
        =
        \Omega\!\left(\frac{n^{1/m}}{(\log n)^{4/m}}\right).
\]

It remains to track the weight of the \(Z\)-stabilizer generators. The concatenated code has
two standard types of \(Z\)-stabilizers.

The first type consists of the inner \(Z\)-stabilizers. These have row weight at most \(\nu\).

The second type consists of lifts of outer \(Z\)-stabilizers. Let \(v=(v_i)\in\K^E\) be an outer
\(Z\)-stabilizer. Fix an information set \(I_{\mathrm{in}}\) of size \(\kappa\) for the inner code.
After choosing coordinates on \(I_{\mathrm{in}}\), each element of \(\K\) is identified with a vector
in \(\F_p^\kappa\). For any \(i\), multiplication by \(v_i\) is represented by a
\(\kappa\times \kappa\) matrix over \(\F_p\) acting on the information set of the \(i\)-th inner code block. Thus the single \(\K\)-linear constraint
\(\sum_i v_i z_i=0\) gives \(\kappa\) scalar \(Z\)-checks over \(\F_p\), and each such scalar check
uses at most \(\kappa\) positions inside each nonzero outer block. Therefore an outer
\(Z\)-stabilizer of weight \(w\) lifts to \(\F_p\)-ary \(Z\)-stabilizers of row weight at most
\(\kappa w\). Since the outer \(Z\)-stabilizer generators have  weight \(O(N^{1/m})\), their lifts have
row weight \(O(\kappa N^{1/m})\).

Next consider the restriction step. Recall that the outer CSS code from Theorem~\ref{thm:main}
has logical space \(\K^A\), and in the alphabet-reduction step we restrict this logical space
to \(S_\tau^A\). This adds \(Z\)-stabilizers corresponding to the linear constraints that vanish on
\(S_\tau^A\). Let
\[
        S_\tau^\perp
        =
        \{\lambda\in\K:
        \operatorname{Tr}_{\K/\F_p}(\lambda s)=0
        \text{ for every }s\in S_\tau\}.
\]
Then these added constraints are exactly the constraints 
\(\operatorname{Tr}_{\K/\F_p}(\lambda y_a)=0\), for all \(a\in A\) and
\(\lambda\in S_\tau^\perp\).

We now show that these added stabilizers also have low-weight representatives. In the proof of
Proposition~\ref{prop:Z-low-weight}, for every \(a\in A\) we constructed a  dual codeword
\(h_{H\cdot a,a}\in C_r^\perp\) whose \(a\)-th coordinate is \(1\),  whose remaining support lies in \(E\), and
\(\wt(h_{H\cdot a,a})\le |H|=O(N^{1/m})\).

Let \(c\in C_r\) with   \(c|_A=y\in\K^A\). Since
\(h_{H\cdot a,a}\in C_r^\perp\), we have
\[
        0=\langle h_{H\cdot a,a},c\rangle
        =
        y_a+\langle h_{H\cdot a,a}|_E,c|_E\rangle .
\]
Hence \(-h_{H\cdot a,a}|_E\) represents a functional on the physical space that represents the  logical functional \(y\mapsto y_a\). Hence, by taking the trace and lifting this functional  through the
inner code, we get that its weight is  multiplied by at most \(\kappa\). Thus every added restriction stabilizer has
row weight \(O(\kappa N^{1/m})\).

Altogether, the final \(Z\)-stabilizer generators have weight at most
\(
        O(\max\{\nu,\kappa N^{1/m}\})=O(\kappa N^{1/m}).
\)
 
Since \(n=\nu n_{\rm out}=\Theta(\nu N)=\Theta(\kappa^4 N)\), we have
\[
        \kappa N^{1/m}
        =
        O\!\left(
        \kappa\left(\frac{n}{\kappa^4}\right)^{1/m}
        \right)
        =
        O\!\left(
        n^{1/m}(\log n)^{1-4/m}
        \right),
\]
using \(\kappa=\Theta(\log n)\). Thus the final \(Z\)-stabilizer generators have
weight
\[
        O\!\left(
        n^{1/m}(\log n)^{1-4/m}
        \right).
\]
        
\end{proof}

\section{Discussion and open problems}
\label{sec:discussion}
We have shown that the algebraic puncturing framework for transversal \(CCZ\) codes can be combined with algebraic expander codes without losing the linear number of logical qudits. A direct use of the Golowich--Guruswami puncturing theorem runs into the small dual distance of these codes, and therefore yields only a sublinear number of logical qudits. Our construction shows that this loss is not inherent to the algebraic expander codes themselves. The key step is to choose the puncturing set as a large interpolation set for a lower-rate algebraic expander code, which gives a puncturing statement tailored to this family. As a result, the growing-alphabet construction has linear dimension, transversal \(CCZ\) gates, and asymmetric distances with one distance linear and the other polynomial.

A separate feature of the construction is that the \(Z\)-stabilizer space is generated by explicit low-weight checks. In particular, the \(Z\)-stabilizer space admits an explicit generating set of sublinear weight. This shows that one does not have to leave the Schur-product puncturing framework in order to obtain sublinear stabilizer locality, at least on one side. The one-shot alphabet reduction then carries these features to fixed prime fields, with near-linear dimension, an asymmetric distance profile, and \(Z\)-stabilizer locality preserved up to polylogarithmic factors.

Several questions remain open, and we mention two that seem most immediate. The most natural is whether the inherent asymmetry of the construction can be removed. In particular, can one obtain both  distances linear in the block length while also admitting low-weight generating sets for both the \(X\)- and \(Z\)-stabilizer spaces? In the present construction, only the \(Z\)-stabilizer space has an explicit sublinear-weight generating set, and only the \(X\)-distance is linear. Removing this asymmetry would bring the construction closer to a genuinely low-density family of codes with transversal CCZ.

A second question concerns decoding. The construction has both Tanner and algebraic structure, but we do not give an efficient decoder for the resulting quantum codes. It would be useful to understand whether these two structures can be combined to give an efficient decoder.

\section*{Acknowledgments}

The authors thank Louis Golowich for suggesting the comparison with the direct-sum construction based on quantum RS codes, which led to a clearer presentation of the contribution of this work.
\bibliographystyle{alpha}
\bibliography{biblio}

\appendix

\section{Proof of Proposition \ref{prop:projective-multiplicity-mfc}}
\label{app:projective-multiplicity-mfc}

\begin{proof}
Fix an irreducible polynomial \(\gamma(X)\in\mathbb F_q[X]\) of degree \(\kappa\), and identify
\(\mathbb F_{q^\kappa}\cong \mathbb F_q[X]/(\gamma(X))\).
Thus every element of \(\mathbb F_{q^\kappa}\) has a unique representative in
\(\mathbb F_q[X]^{<\kappa}\).

 For \(b\in\mathbb F_q\), the Hasse derivatives
\(\partial_b^{(i)}f\) are defined by the expansion
\[
f(X)=\sum_{i\ge 0}\partial_b^{(i)}f\,(X-b)^i .
\]
For the point at infinity, with respect to a degree bound \(B\), define
\(\partial_{\infty,B}^{(i)}f=[X^{B-i}]f\), the coefficient of \(X^{B-i}\).

Let \(\mathcal T\) be the set of tuples \((P;i_1,\ldots,i_a)\) such that
\(P\in\mathbb P^1(\mathbb F_q)\), \(0\le i_j< \kappa\), and \(i_1+\cdots+i_a<\ell_P\).
Thus \(|\mathcal T|=\nu\).

For each \(h\in[a]\), define an \(\mathbb F_q\)-linear encoder
\(\Enc^{(h)}:\mathbb F_{q^\kappa}\to\mathbb F_q^{\mathcal T}\) as follows. Given
\(z\in\mathbb F_{q^\kappa}\), let \(f\in\mathbb F_q[X]^{<\kappa}\) be its
unique representative modulo \(\gamma(X)\), and set
\[
\Enc^{(h)}(z)_{(P;i_1,\ldots,i_a)}=
\begin{cases}
\partial_P^{(i_h)}f, & P\in\mathbb F_q,\\
\partial_{\infty,\kappa-1}^{(i_h)}f, & P=\infty.
\end{cases}
\]
and let \(C^{(h)}=\operatorname{im}(\Enc^{(h)})\).

We first show that \(\Enc^{(h)}\) is injective. Suppose \(\Enc^{(h)}(z)=0\), and let \(f\in\mathbb F_q[X]^{<\kappa}\) be the representative of \(z\).
For each \(P\in\mathbb P^1(\mathbb F_q)\) and each \(0\le r<\min\{\ell_P,\kappa\}\),
the tuple with \(i_h=r\) and all other \(i_j=0\) lies in \(\mathcal T\), so all such Hasse derivatives vanish.
Set \(t_P=\min\{\ell_P,\kappa\}\). Then \(\sum_P t_P\ge \kappa\). Indeed, if at least one \(t_P\geq \kappa\) the inequality follows. Otherwise \(\sum_{P}t_P=\sum_{P}\ell_P=a(\kappa-1)+1\geq \kappa\).

 The vanishing at infinity implies
\(\deg f\le \kappa-1-t_\infty\), while the finite-point conditions imply that the total affine zero multiplicity of \(f\) is at least
\[
    \sum_{b\in\mathbb F_q} t_b
    \ge \kappa-t_\infty .
\]
This is impossible unless \(f=0\). Thus \(\Enc^{(h)}\) is injective and \(\dim C^{(h)}=\kappa\).

We now define the decoder. Given \(w\in\mathbb F_q^{\mathcal T}\), for each
\(P\in\mathbb P^1(\mathbb F_q)\) and \(0\le r<\ell_P\), define
\[
y_{P,r}=
\sum_{\substack{i_1+\cdots+i_a=r\\0\le i_j< \kappa}}
w_{(P;i_1,\ldots,i_a)}.
\]
We claim there is a unique polynomial \(W\in\mathbb F_q[X]\) of degree at most \(D=a(\kappa-1)\) such that
\[
    \partial_b^{(r)}W=y_{b,r}
    \quad\text{for all } b\in\mathbb F_q,\ 0\le r<\ell_b,
\]
and
\[
    \partial_{\infty,D}^{(r)}W=[X^{D-r}]W=y_{\infty,r}
    \quad\text{for }0\le r<\ell_\infty .
\]

We prove uniqueness first. The interpolation conditions define a linear map from
the space of polynomials of degree at most \(D\) to
\(\mathbb F_q^{D+1}\). Thus it suffices to show that the only polynomial
satisfying the same conditions with all prescribed values equal to zero is the
zero polynomial. Suppose that \(W\in\mathbb F_q[X]\) has degree at most
\(D\) and satisfies the above interpolation conditions with all prescribed
values equal to zero. The conditions at infinity imply   that
\(
    \deg W\le D-\ell_\infty .
\)
On the other hand, for each \(b\in\mathbb F_q\), the conditions 
\(
    \partial_b^{(r)}W=0
    \)for every \(0\le r<\ell_b 
\),
imply that  \(W\) has a zero of multiplicity at least \(\ell_b\) at \(b\).
Hence the total multiplicity of the affine roots of \(W\) is at least
\(
    \sum_{b\in\mathbb F_q}\ell_b
    =
    D+1-\ell_\infty .
\)
This is strictly larger than the degree bound \(D-\ell_\infty\), so \(W=0\).
Therefore there is at most one polynomial satisfying the interpolation
conditions. Since the space of polynomials of degree at most \(D\) has dimension
\(D+1\), and the number of prescribed values is
\(
    \sum_{P\in\mathbb P^1(\mathbb F_q)}\ell_P=D+1,
\)
the mapping is bijective, and existence follows as well. Thus there is a unique such polynomial \(W\).

Define \(\Dec(w)=W\bmod \gamma(X)\in\mathbb F_{q^\kappa}\).
This map is \(\mathbb F_q\)-linear.

It remains to prove multiplication-friendliness. Let \(z_h\in\mathbb F_{q^\kappa}\) with representatives
\(f_h\in\mathbb F_q[X]^{<\kappa}\), and set \(w=\Enc^{(1)}(z_1)*\cdots*\Enc^{(a)}(z_a)\).

For \(b\in\mathbb F_q\) and \(0\le r<\ell_b\), the Hasse product rule gives\[
\partial_b^{(r)}(f_1\cdots f_a)
=
\sum_{\substack{i_1+\cdots+i_a=r\\0\le i_j\le \kappa-1}}
\partial_b^{(i_1)}f_1\cdots \partial_b^{(i_a)}f_a
=
y_{b,r}.
\]

For \(0\le r<\ell_\infty\), at infinity, using \(D=a(\kappa-1)\), we have
\[
\partial_{\infty,D}^{(r)}(f_1\cdots f_a)
=
[X^{D-r}](f_1\cdots f_a)
=
\sum_{\substack{i_1+\cdots+i_a=r\\0\le i_j<\kappa}}
\prod_{j=1}^a [X^{\kappa-1-i_j}]f_j
=
y_{\infty,r}.
\]
 Since \(\deg(f_1\cdots f_a)\le D\), uniqueness gives
\(W=f_1\cdots f_a\). Therefore
\[
\Dec(w)=W\bmod \gamma(X)=f_1\cdots f_a \bmod \gamma(X)=z_1\cdots z_a,
\]
proving the claim.
\end{proof}

\section*{Author Contributions}
All authors contributed to the research and writing of this manuscript.

During the preparation of this work, the authors used a large language model for language editing to improve the clarity and presentation of the manuscript. The authors take full responsibility for the content of the paper.

\end{document}